\documentclass[conference]{IEEEtran}

\usepackage{amsmath,amsfonts,amssymb,amsthm}
\usepackage{algorithmic}
\usepackage{array}
\usepackage{subcaption}
\usepackage{graphicx}
\usepackage{multirow}
\usepackage{textcomp}
\usepackage{stfloats}
\usepackage{url}
\usepackage{verbatim}
\usepackage{graphicx}
\usepackage{color, soul}
\usepackage{bm}

\usepackage{cite}
\usepackage{xcolor}
\usepackage{xpatch}

\newtheorem{theorem}{Theorem}

\usepackage{algorithmic}
\usepackage[linesnumbered,ruled]{algorithm2e}
\usepackage{cuted}

\def\BibTeX{{\rm B\kern-.05em{\sc i\kern-.025em b}\kern-.08em
    T\kern-.1667em\lower.7ex\hbox{E}\kern-.125emX}}
\usepackage{balance}

\begin{document}
% \title{Digital Twin-based 3D Map Management for Edge-assisted Device Pose Tracking in Mobile AR}
% \title{User-centric Service Provisioning for Edge-assisted Mobile AR: A Digital Twin-based Approach}
\title{QoE-oriented Communication Service Provision for Annotation Rendering in Mobile Augmented Reality}

{
    \author{
    \IEEEauthorblockN{Lulu~Sun\IEEEauthorrefmark{1}, Conghao~Zhou\IEEEauthorrefmark{1},~Shisheng~Hu\IEEEauthorrefmark{2}, Yupeng Zhu\IEEEauthorrefmark{1},~Nan~Cheng\IEEEauthorrefmark{1}\IEEEauthorrefmark{3}, and Xu (Tony) Xia\IEEEauthorrefmark{4}
        \IEEEauthorblockA{
          \IEEEauthorrefmark{1}School~of~Telecommunications Engineering,~Xidian University,~China
        \\\IEEEauthorrefmark{2}Department~of~Electrical~and~Computer~Engineering,~University~of~Waterloo,~Canada
        \\\IEEEauthorrefmark{3}The State Key Laboratory of ISN,~Xidian University,~China
        \\\IEEEauthorrefmark{4}China Telecom Research Institute,~China
        \\ \{lulusun, yupeng\}@stu.xidian.edu.cn, \{conghao.zhou dr.nan.cheng\}@ieee.org, s97hu@uwaterloo.ca, xiaxu@chinatelecom.cn}
            }
}

\maketitle

\begin{abstract}

As mobile augmented reality (MAR) continues to evolve, future 6G networks will play a pivotal role in supporting immersive and personalized user experiences. In this paper, we address the communication service provision problem for annotation rendering in edge-assisted MAR, with the objective of optimizing spectrum resource usage while ensuring the required quality of experience (QoE) for MAR users. To overcome the challenges of user-specific uplink data traffic patterns and the complex operational mechanisms of annotation rendering, we propose a digital twin (DT)-based approach. We first design a DT specifically tailored for MAR applications to learn key annotation rendering mechanisms, enabling the network controller to access MAR application-specific information. Then, we develop a DT-based QoE modeling approach to capture the unique relationship between individual user QoE and spectrum resource demands. Finally, we propose a QoE-oriented resource allocation algorithm that decreases resource usage compared to conventional network slicing-based approaches. Simulation results demonstrate that our DT-based approach outperforms benchmark approaches in the accuracy and granularity of QoE modeling.

\end{abstract}

% \begin{IEEEkeywords}
% Augmented reality, annotation rendering, service provision, 6G, digital twin. 
% \end{IEEEkeywords}

\section{Introduction}

The international telecommunication union has published the ``IMT-2030~(6G) Framework'' to guide the standardization of future 6G communication networks, envisioning immersive communications as an extension of existing enhanced Mobile Broadband (eMBB) communications~\cite{cheng2024toward}. Mobile augmented reality (MAR) as a representative form of immersive communications, aims to provide users with rich and immersive experiences by blending physical and virtual worlds on portable or mobile devices,~e.g.,~smart glasses~\cite{chen2023adaptslam}. To achieve this, a key procedure involves annotation rendering, which entails superimposing virtual content onto the real-world environment as viewed through an MAR device. Currently, the widespread implementation of annotation rendering in MAR still heavily relies on substantial communication network supports~\cite{zhou2024user_wcm}. Specifically, annotation rendering requires real-time estimation of MAR device pose to ensure virtual content are accurately placed from the correct location and perspective relative to the user's viewpoint. This necessitates substantial computing and data storage resources~\cite{huzaifa2021illixr,ma2023nomore,hu2023adaptive}. Therefore, both academia and industry have concentrated on the edge-assisted MAR paradigm, leveraging the adequate resources of edge servers through communication networks. 

Existing communication networks encounter two key challenges of network resource management in supporting annotation rendering for edge-assisted MAR. First, each MAR device uploads its captured camera frames, to an edge server for device pose estimation and the delivery of the required virtual content. Since the real-world environment and human movement patterns, such as head turning, significantly vary for each user, the number of camera frames required to support annotation rendering in edge-assisted MAR is user-specific~\cite{ran2020multi}. However, in the service-oriented, e.g., slicing-based, network architecture currently employed in 5G networks, differences in data traffic patterns among individual devices running the same MAR application cannot be effectively distinguished. As a result, service-oriented network resource management approaches may reduce the level of immersion of certain users in MAR annotation rendering~\cite{zhou2024digital}. Second, current MAR service provisioning from a quality of service (QoS) perspective, emphasizing network-related factors including data rate. However, beyond conventional network-related factors, numerous application-related factors can influence user experience in MAR applications. Specifically, even with identical network resource allocation, users may experience markedly different levels of immersion, as MAR applications employ different operational mechanisms in annotation rendering to accommodate diverse human behaviors and real-world environments~\cite{linowes2017augmented}. Therefore, conventional QoS-oriented approaches fail to accurately capture the demands of individual MAR users through comprehensively analyzing both network-related and application-related factors.  

Recent studies have explored to support extended reality (XR), an umbrella term encompassing MAR, from the networking and communication perspective~\cite{zhou2024user_wcm}. Zhou~\emph{et~al.} have investigated user-centric communication service provision to satisfy the delay requirements necessary for enabling timely device pose tracking in edge-assisted MAR~\cite{zhou2024digital}. Zhao~\emph{et~al.} have investigated data packet scheduling for multiple XR users, focusing on ensuring the probability of successful transmission within strict delay constraints. While these studies have accounted for QoS differences across users, the optimized QoS metrics often fail to accurately capture user satisfaction in MAR. A few works have paid efforts on QoE modeling and QoE-oriented communication service provision for XR video transmission~\cite{pan2024quality,feng2023qoe}. However, the impact of annotation rendering mechanism in MAR is overlooked in existing studies on communication service provision.

% Bojovic~\emph{et~al.}~\cite{bojovic2023enhancing} have proposed a mechanism to adapt XR video codec to dynamic wireless channel conditions to satisfy QoS requirements. Meanwhile, some works have paid efforts on QoE modeling and QoE-oriented communication resource management for XR video transmission~\cite{pan2024quality,feng2023qoe}.

% Manjunath~\emph{et~al.} have proposed a machine learning-based approach to classify different XR applications according to their associated data traffic patterns~\cite{manjunath2024discern}. 

In this paper, we investigate a communication service provision problem for annotation rendering in edge-assisted MAR. Our objective is to minimize the radio spectrum resource consumed by multiple MAR devices while meeting their requirements for the virtual content hit rate (VCHR), a QoE metric employed for annotation rendering~\cite{ren2020edge}. Different from existing works, we propose the establishment of a digital twin (DT) for each individual MAR device. A DT captures the unique relationship between user QoE and network resource demands. The main contributions of this paper are as follows:
    \begin{itemize}
        \item We define a function of the DT, specifically designed for MAR, to replicate key operational mechanisms of annotation rendering, which enables the network controller to access MAR application-specific information from the network side.

        \item We define a function of the DT to capture the unique impact of radio spectrum resource management on individual user's VCHR, which offers enhanced granularity in QoE modeling compared to conventional user-agnostic approaches.

        \item We develop a QoE-oriented communication service provision algorithm, which enhances radio spectrum resource efficiency compared to the existing slicing-based service provision used for 5G.  

        % \item We establish a personalized hierarchical data model, organizing data attributes carefully chosen for MAR, \bl{to capture the implicit impact of the MAR operational mechanism on the uplink data traffic of an MAR user.} 

        % \item We propose two machine learning-based methods with different complexities for data traffic modeling. In addition, we design an easy-to-use mechanism \bl{for switching between the two methods} to adapt to non-stationary uplink data traffic in MAR.  

        % \item We derive a closed-form resource reservation solution to a service provision problem for an individual MAR device, considering potential inaccuracies in the data-driven traffic modeling, which enhances the robustness of the DT-based service provision approach.
    \end{itemize}

\section{System Model and Problem Formulation}

\subsection{Considered Scenario}

    \begin{figure}[t]
        \centering
        \includegraphics[width=0.32\textwidth]{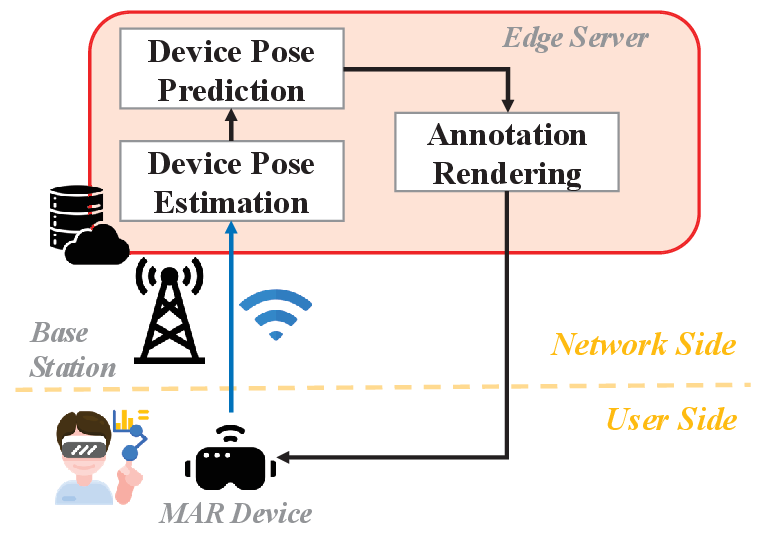}
        \caption{The considered scenario of edge-assisted MAR.}\label{scenario}
        % \vspace{-0.5cm}
    \end{figure}

We consider a scenario where multiple users, utilizing their respective MAR devices to run MAR applications, move freely within the communication coverage of a base station (BS). Each MAR device captures camera frames periodically with a fixed frame rate while rendering 3D virtual content within the MAR user’s field of view~\cite{campos2021orb}. An edge server is deployed at the BS to ensure accurate annotation rendering, providing MAR users with immersive experiences. As shown in Fig.~\ref{scenario}, the edge server is responsible for conducting three key modules to support MAR:
\begin{itemize}
    \item \textbf{Device pose estimation --} The edge server supports estimating the current pose of each MAR device based on the captured camera frames. Each MAR device periodically uploads camera frames along with related data, e.g., timestamps, to the edge server. The edge server then estimates the device poses corresponding to the uploaded camera frame based on the uploaded data~\cite{zhou2024digital,chen2023adaptslam}.  

    \item \textbf{Device pose prediction --} The edge server predicts the future poses of each MAR device based on the history trajectory of its estimated device poses. The predicted device poses are used to determine the possible future viewpoints of MAR users for proactively preparing virtual content~\cite{han2020vivo}.  

    \item \textbf{Annotation rendering --} The edge server delivers a portion of virtual content within each MAR user's field of view to their respective MAR device, where it is rendered by overlaying the virtual content onto the camera frame newly captured by the MAR device.

\end{itemize}

Denote the set of MAR devices by~$\mathcal{U}$, and let~$\mathcal{F}_{u}$ denote the set of camera frames captured by MAR device~$u \in \mathcal{U}$.

\subsection{Communication Model} 

Timely uploading camera frames to the edge server requires proper radio spectrum resource allocation. Let~$r_{u}$ denote the uplink data rate of MAR device~$u$ when camera frame~$f$ is captured, given by:
    \begin{equation}\label{eq1}
        r_{u,f} = b_{u} \log (1 + \gamma_{u,f}), \,\, \forall f \in \mathcal{F}_{u}, u \in \mathcal{U},
    \end{equation}
where~$b_{u}$ denotes the amount of radio spectrum resource allocated to MAR device~$n$ for uplink communication, and~$\gamma_{u,f}$ represents a random variable corresponding to the signal-to-noise ratio experienced by MAR device~$u$ when camera frame~$f$ is captured. Let~$\alpha$ denote the volume (in bits) of data uploaded to the edge server for estimating the device pose corresponding to a camera frame, assuming an identical data volume for device pose estimation across all camera frames. We denote by~$S_{u,f}$ the duration of uplink transmission for camera frame~$f$ at MAR device~$u$, and the average uplink transmission latency of MAR device~$u$ for a camera frame, denoted by~$\mathbb{E}[S_{u}]$, is as follows: 
    \begin{equation}\label{}
        \mathbb{E}[S_{u}] = \mathbb{E}_{f} \left[\frac{\alpha}{r_{u,f}} \right],\,\, \forall u \in \mathcal{U},
    \end{equation}
where the signal-to-noise ratio remains constant during data transmission for each camera frame but may vary between different camera frames.

Due to the limited radio spectrum and computing resources for estimating the device pose corresponding to each camera frame, a portion of all captured camera frames can be selected to be uploaded the edge server~\cite{chen2023adaptslam}. Considering that each MAR device captures camera frames periodically, we assume that the camera frames to be uploaded by each MAR device are also selected uniformly. Let~$\mathcal{F}^\text{s}_{u} \subseteq \mathcal{F}_{u}$ denote the set of camera frames uploaded by MAR device~$u$, and let~$\lambda_{u}$ represent the upload frequency (in Hz) for the data associated with the set of camera frames~$\mathcal{F}^\text{s}_{u}$. Based on the D/G/1 queuing model, the average latency from when a camera frame is selected by MAR device~$n$ to when it is fully uploaded, denoted by~$\tau_{u}$, satisfies the following constraint:
    \begin{equation}\label{}
        \tau_{u} \approx \frac{\lambda_{u} \mathbb{E}[S^{2}_{u}] }{ 2(1 - \lambda_{u}\mathbb{E}[S_{u}]) } \leq T,\,\, \forall u \in \mathcal{U},
    \end{equation}
where~$T$ denotes the maximum tolerable latency for uploading data associated with each camera frame from an MAR device.

\subsection{Annotation Rendering Model}

In MAR applications, overlaying virtual content onto the physical environment, i.e., a camera frame, relies on the viewpoint of users, which is determined by the 6\,DoF pose of their respective MAR devices. Let $\mathbf{q}_{u, f} = [t^{x}_{u,f}, t^{y}_{u,f}, t^{z}_{u,f}, \theta^{x}_{u,f}, \theta^{y}_{u,f}, \theta^{z}_{u,f}]^{\top}$ denote the estimated 6\,DoF pose of MAR device~$u$ when camera frame~$f$ is captured, comprising 3\,DoF for translational movement and 3\,DoF for rotational movement (i.e., yaw, pitch, and roll)~\cite{ran2020multi,chen2023adaptslam}. Since the annotation rendering of 3D virtual content is time-consuming, a device pose needs to be predicted proactively. Denote~$\hat{\mathbf{q}}_{u, f+W}$ the predicted pose of MAR device~$n$ when camera frame~$f+W$ is captured, given by:  
    \begin{equation}\label{eq4}
        \hat{\mathbf{q}}_{u, f+W} = G_{u}(\mathcal{Q}_{u,f}), \,\,\,\, \forall u \in \mathcal{U},
    \end{equation}
where the set of input device poses~$\mathcal{Q}_{u}$ that are selected by down-sampling from the original camera frame sequence when~camera frame~$f \in \mathcal{F}$ is captured, is as follows:
    \begin{equation}\label{}
        \mathcal{Q}_{u,f} = \left\{ \mathbf{q}_{u, k} | f-H < k \lambda_{u} \leq f, \forall k \in \mathbb{Z}^{+} \right\}.
    \end{equation}
Terms~$W$ and~$H$ denote the lookahead and history windows, respectively, and the function~$G_{u}(\cdot)$ represents a device pose prediction model, e.g., a deep neural network, for MAR device~$u$, which takes as input the estimated device poses corresponding to camera frames from~$f-H$ to~$f$.  

The 3D virtual content is spatially segmented into smaller 3D \emph{cells}, each individually encoded and capable of being fetched separately for delivery and rendering~\cite{han2020vivo}. Given a predicted device pose~$\hat{\mathbf{q}}_{u, f}$, the set of 3D cells required for annotation rendering can be determined, which is denoted by~$\mathcal{C}(\hat{\mathbf{q}}_{u, f})$. Due to the device pose prediction error, the virtual content rendered on the MAR device and the 3D virtual content may differ from the 3D virtual content that aligns with the MAR user's actual viewpoint. To measure such a difference for camera frame~$f$, the virtual content hit rate (VCHR) of MAR device~$u$ is defined as follows:
    \begin{equation}\label{}
        h_{u,f} = \frac{|\mathcal{C}(\mathbf{q}_{u, f}) \cap \mathcal{C}(\hat{\mathbf{q}}_{u, f})|}{|\mathcal{C}(\mathbf{q}_{u, f}) \cup \mathcal{C}(\hat{\mathbf{q}}_{u, f})|}, \,\, \forall \,\, \mathcal{C}(\mathbf{q}_{u, f}) \cup \mathcal{C}(\hat{\mathbf{q}}_{u, f}) \neq \emptyset,
    \end{equation} 
where $\cap$ and $\cup$ denote the intersection and the union of two sets, respectively, and~$|\cdot|$ represents the cardinality of a set. The smaller the value of~$h_{u,f}$, the greater the deviation of the rendered virtual content from an MAR user's actual viewpoint, thereby diminishing the user's immersive experiences. 

Let~$\mathcal{F}^\text{r}_{u}$ denote the set of camera frames requiring annotation rendering at MAR device~$u$. To ensure the quality of experience (QoE) in annotation rendering, the VCHR of MAR device~$u \in \mathcal{U}$ needs to satisfy the following probabilistic constraint: 

    \begin{equation}\label{eq7}
        P \left( \left( \frac{1}{|\mathcal{F}^\text{r}_{u}|} \sum_{f \in \mathcal{F}^\text{r}_{u}}{\mathbf{1}\left( h_{u,f}  \ge V_{u} \right)}  \right) \ge \rho \right) \ge \varepsilon,
        % P ( V \ge h_{u,f} ) \ge \varepsilon, \,\,\,\, \forall f \in \mathcal{F}_{u}, u \in \mathcal{U},
    \end{equation}
where~$V_{u}$ represents the maximum tolerable VCHR experienced by the user through MAR device~$u$ during annotation rendering,~$\mathbf{1}(\cdot)$ is an indicator function,~$\rho$ denotes the ratio of camera frames satisfying the VCHR requirement to the total number of rendered camera frames, and $\varepsilon \in [0, 1]$ represents the required reliability in communication service provision.

\subsection{Problem Formulation}

From a communication and networking perspective, we formulate a QoE-aware communication service provision problem for MAR applications with the objective of minimizing the usage of radio spectrum resource allocated to MAR devices, as follows:

    \begin{subequations}\label{p1}
        \begin{align}
            \textrm{P1:} &\,\, \min_{ \{ b_{u} \}_{u \in \mathcal{U}} } \sum_{ u \in \mathcal{U} }{ b_{u} }\\
            \textrm{s.t.} &\,\, \eqref{eq1} - \eqref{eq7},\\
            & \,\, \forall b_{u} \in \mathbb{R},
        \end{align}
    \end{subequations}
where the optimization variable~$b_{u}$ represents the amount of the radio spectrum resource allocated for MAR devices, and~\eqref{eq7} serves as the primary constraint to ensure the QoE,~i.e., VCHR, for individual MAR devices. 

Problem~P1 is intractable in existing communication networks for two reasons. First, from the communication and networking perspective, the network controller currently lacks the MAR application-specific information necessary for QoE-aware network resource management stemming from the separation between network management and application development domains. As a result, the quantitative impact of radio spectrum resource allocation decisions, i.e.,~$b_{u}$, on the QoE value is unknown~\emph{a priori}. Second, the limited data analysis capabilities of current communication networks pose challenges to differentiating QoE models for individual MAR users. Thus, mathematically satisfying the QoE constraints~\eqref{eq7} for each MAR user through network resource optimization becomes infeasible. 
% To address these two challenges, we develop digital user agents for MAR users, which enable user-specific QoE modeling in communication networks and thus support QoE-aware communication service provisioning.  

\section{The Developed Approach}

In this section, we develop a novel approach to address the aforementioned two challenges.  

\subsection{Developed Digital Twin}

To enable QoE-oriented communication service provision (QoE-CSP) for MAR, the network controller requires new network functionalities: 1) a capability to provide MAR application-related information, allowing the network controller to account for the impact of MAR operational mechanisms in network resource management; and 2) a capability to differentiate the service demands of individual MAR users through advanced user data analysis, enabling personalized network resource management. To this end, we introduce an entity called the digital twin (DT), which is established for each individual MAR device. Serving as a proxy~\cite{ma2024from,song2024canal}, the DT can provide the network controller with application-related and user-specific information about the respective MAR device, thereby supporting QoE-CSP. Specifically, we design the following two functions for each DT, each computationally providing the required application-related and user-specific information~\cite{zhou2024user_wcm,cheng2024toward}.    

\subsubsection{Application-related information extraction}

As mentioned in Subsection~II.C, the predicted poses~$\hat{\mathbf{q}}_{u, f+W}$, obtained via the predictor $G(\cdot)$, are crucial for annotation rendering in MAR, it is essential to incorporate the MAR operational mechanism, more specifically the details of $G(\cdot)$, into communication service provision to ensure that the QoE requirements of MAR users are met. Since MAR application-related information may not accessible to the network controller, the specific details of the predictor $G(\cdot)$ used in~\eqref{eq4} may not be publicly available for radio spectrum resource management. 

This function involves replicating the MAR application-related behavior from a networking and communication perspective. Instead of focusing on the human behavior of MAR users, it clones the behavior associated with the operational mechanisms of the MAR application. Define~$\hat{G}(\cdot; \boldsymbol{\vartheta})$ as the cloned device pose predictor parameterized by~$\boldsymbol{\vartheta}$, which is obtained by minimizing the following loss function:  
    \begin{equation}\label{}
        L(\boldsymbol{\vartheta}) = \frac{1}{|\Xi|} \sum_{(\hat{\mathbf{q}}_{u, f+W},\mathcal{Q}_{u,f}) \in \Xi} \left(\hat{\mathbf{q}}_{u, f+W} - \hat{G}(\mathcal{Q}_{u,f}; \boldsymbol{\vartheta}) \right)^{2},
    \end{equation}
where~$\Xi$ represents the set of device pose trajectories for all MAR devices. The core of this function is to accurately replicate the performance of device pose prediction as it is actually utilized in MAR, rather than focusing on improving pose prediction accuracy. As suggested in~\cite{han2020vivo}, a linear regression technique is employed to construct the predictor~$G$ for annotation rendering, we adopt the same technique for cloning. If specific details about the employed technique are unavailable, $\hat{G}(\cdot; \boldsymbol{\vartheta})$ can be approximated using more general machine learning methods.   

\subsubsection{User QoE modeling}

To capture the unique impact of a radio spectrum resource management decision on the QoE value, i.e., VCHR, we define a user-specific function as~$\Omega_{u} (\cdot)$. Given the input of $b_{u}$ and~$\hat{\mathbf{q}}_{u, f}$, the probability of $h_{u,f} = y, \forall y \in \left[0,1\right]$, is estimated as follows:
    \begin{equation}\label{}
        \hat{P}(h_{u,f} = y | b_{u}, \hat{\mathbf{q}}_{u, f}) =  \Omega_{u} \left(b_{u}, \hat{\mathbf{q}}_{u, f}; \boldsymbol{\theta}_{u} \right), \,\,\,\, \forall f \in \mathcal{F}_{u}, u \in \mathcal{U},
    \end{equation}
where $\sum_{\forall y}{\hat{P}(h_{u,f} = y | b_{u}, \hat{\mathbf{q}}_{u, f})} = 1$, and~$\boldsymbol{\theta}_{u}$ represents the parameters of the mapping function~$\Omega_{u} (\cdot)$ corresponding to MAR device~$u$. Given function~$\Omega_{u} (\cdot)$, the probability of camera frame~$f \in \mathcal{F}^\text{r}_{u}$ that is rendered to meet the VCHR requirement, can be estimated as follows:
   \begin{subequations}\label{eq10}
        \begin{align}
            \hat{p}_{u,f} & = \int _{y \ge V_{u}}{\hat{P}(h_{u,f} = y | b_{u}, \hat{\mathbf{q}}_{u, f})}\\
            & = \int_{y \ge V_{u}, y \in \mathbb{R}}{\Omega_{u} \left(b_{u}, \hat{\mathbf{q}}_{u, f}; \boldsymbol{\theta}_{u} \right)},\\
            & \approx \sum_{y \ge V_{u}, y \in \mathbb{Z}}{\Omega_{u} \left(b_{u}, \hat{\mathbf{q}}_{u, f}; \boldsymbol{\theta}_{u} \right)},
        \end{align}
    \end{subequations}
where we output a probability mass function of~$\hat{P}(h_{u,f} = y | b_{u}, \hat{\mathbf{q}}_{u, f})$,~as shown in~(\ref{eq10}c), by utilizing a deep neural network with a soft-max activation function. Since QoE requirements, i.e.,~$V_{u}$, may vary across MAR devices, the function~$\Omega_{u} (\cdot;  \boldsymbol{\theta}_{u})$ is uniquely trained for each MAR device using its respective user data such as device pose trajectories.

\subsection{QoE-oriented Communication Service Provision}

Based on the two functions involved in each DUA, we can quantitatively capture the relationship between radio spectrum resource management and user QoE. The following Theorem~\ref{theorem1} allows us to re-write~\eqref{eq7} as~\eqref{eq11} when the number of camera frames required for annotation rendering is large. 
    \begin{theorem}\label{theorem1}
    For MAR decvie~$u$, given the estimated probability~$\hat{p}_{u,f}$, QoE constraint~\eqref{eq7} can be re-written as follows: 
        \begin{equation}\label{eq11}
            \frac{|\mathcal{F}^\text{r}_{u}| \rho -  \sum_{f \in \mathcal{F}^\text{r}_{u}}{\hat{p}_{u,f}} }{ \sqrt{\sum_{f \in \mathcal{F}^\text{r}_{u}}{\hat{p}_{u,f} (1 - \hat{p}_{u,f})}} } \leq \Phi^{-1}(1-\varepsilon), 
        \end{equation}
    where~$\Phi^{-1}$ is the inverse standard normal cumulative distribution function.
    \end{theorem}
    \begin{proof}   
        See Appendix~A.
    \end{proof}
Based on Theorem~\ref{theorem1}, we can leverage the application-related and user-specific information provided by each DUA to enable QoE-CSP for individual MAR devices. As shown in Algorithm~\ref{alg1}, we propose a QoE-CSP algorithm to address Problem~P1. In Lines~3 to~8, for each camera frame required for annotation rendering on MAR device~$u$, i.e.,~$\mathcal{F}^\text{r}_{u}$, we predict the corresponding device pose and calculate the associated probability~$\hat{p}_{u,k}$ using~\eqref{eq10}. In Lines~9 to~16, we heuristically determine the minimum amount of radio spectrum resource required by each MAR device, i.e.,~$b^{*}_{u}$, to ensure the timely uploading of sufficient camera frames to the edge server. This heuristic method is based on the observation that an increase in radio spectrum resources positively impacts the performance of device pose prediction and, consequently, the VCHR, although the extent of this impact may vary across different MAR devices. Apart from a potential constraint on the total available radio spectrum resources for concurrent uplink transmissions at the BS, the communication service provision for different MAR devices in Problem~P1 is independent. Thus, by using Algorithm~\ref{alg1}, we can determine the minimum radio spectrum resource for all MAR devices while meeting their respective QoE requirements.

    \begin{algorithm}[t] 
        \caption{QoE-CSP Algorithm}\label{alg1}
        \LinesNumbered
        \textbf{Input:} $\Omega(\cdot;\boldsymbol{\theta}_{u})$, $\hat{G}(\cdot)$, $\mathcal{F}^\text{r}_{u}, \delta$;\\
        \textbf{Initialization:} $k = K$, $b_{u} = b^\text{min}_{u}$;\\
        $\mathcal{Q}_{u,k}$ $\leftarrow$ prepare the input data used for predicting the device pose corresponding to camera frame~$k$;\\

        \For{$k \in \mathcal{F}^\text{r}_{u}$}
        {   

            $\hat{\mathbf{q}}_{u, k}$ $\leftarrow$ Predict the device pose corresponding to camera frame~$k$ based on predictor~$\hat{G}(\mathcal{Q}_{u,k})$;\\

            $\hat{p}_{u,k}$ $\leftarrow$ Calculate the probability using~\eqref{eq10};\\

            $\mathcal{Q}_{u,k'}$ $\leftarrow$ prepare the input data for camera frame~$k' \in \mathcal{F}^\text{r}_{u}$;\\

        }

        \For{$b_{u} \leq b^\text{max}_{u} $}
        {   

            \eIf{\eqref{eq11} is satisfied}
                {   
                    $b^{*}_{u} \leftarrow b_{u}$;\\
                    \textbf{Break};\\
                }
                {   
                    $b_{u} \leftarrow b_{u} + \delta$;\\
                    
                }

        }
        \textbf{Output:} $b^{*}_{u}$

    \end{algorithm}

\section{Performance Evaluation}

\subsection{Simulation Settings}

We use a public dataset (\url{https://github.com/Yong-Chen94/6DoF_Video_FoV_Dataset}) that records the pose traces of $40$ participants moving with six degrees of freedom while watching a volumetric video titled \textit{Longdress}, which is overlaid at a fixed location in the physical world. The video lasts $10$ seconds and has a fraqme rate of $30$ frames per second. Each frame is represented by a point cloud, which is equally segmented into $4\times4\times2$ cells. 
% At each pose collection frequency, as suggested in~\cite{han2020vivo}, the linear regression method is used for user pose prediction. Specifically, the user pose $10$ frames ahead is predicted according to the poses collected in the previous $5$ frames. 
Given the predicted and actual user pose in a frame, the cells that are in the viewing frustum and not occluded by the other cells are determined as the visible cells according to the algorithm developed in~\cite{han2020vivo}. 

\subsection{Performance of the DT-based Approach}
    \begin{figure}[t]
        \centering
        \includegraphics[width=0.3\textwidth]{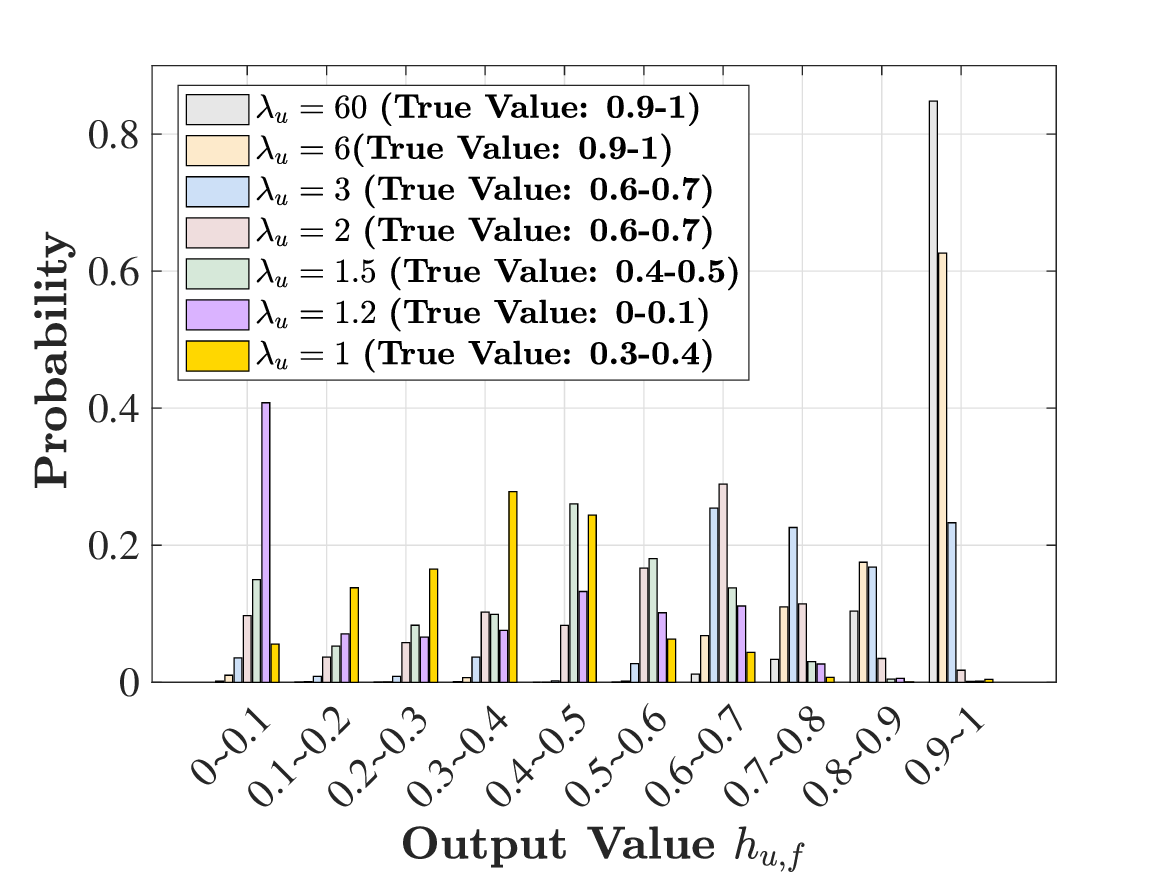}
        \caption{The output of the DT customized for User~\#7.}\label{fig1}
    \end{figure}

In Fig.~\ref{fig1}, we show the output of the DT, i.e., the predicted probability distribution of VCHR, customized from User~\#7 under different upload frequencies. We can observe that the customized DT effectively captures the relationship between upload frequency upload frequency, associated with the allocated radio spectrum, and the VCHR for User~\#7. Notably, as the upload frequency varies, the predicted VCHR shows a general downward trend, aligning closely with the true values. This result shows that the effectiveness of the DT-based approach in modeling QoE for individual users in MAR.

    \begin{figure}[t]
        \centering
        \includegraphics[width=0.3\textwidth]{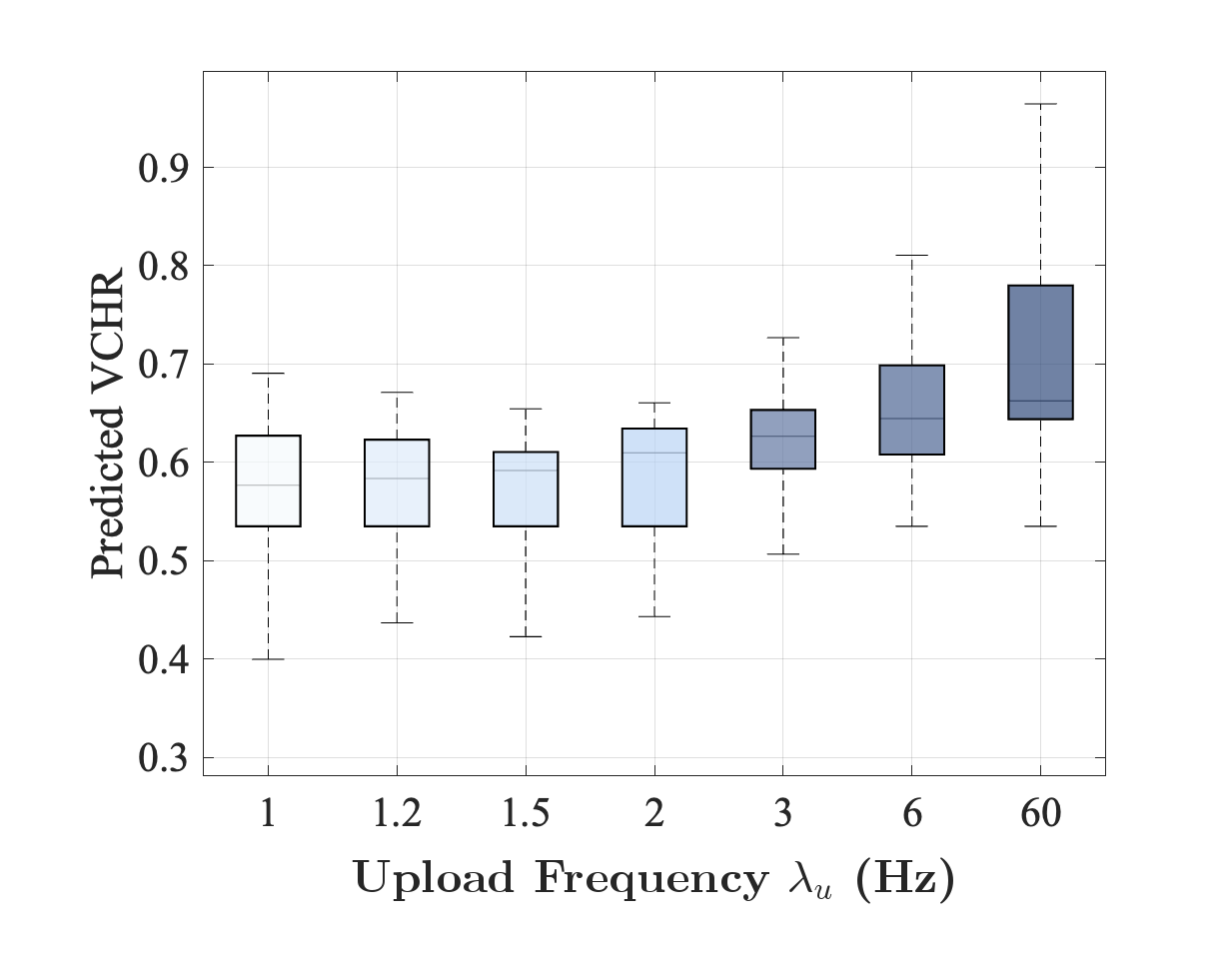}
        \caption{Predicted QoE value, i.e., VCHR, versus upload frequency for User~\#1.}\label{fig3}
    \end{figure}

In Fig.~\ref{fig3}, we show the QoE modeling performance for User~\#1. For each value upload frequency~$\lambda_{u}$, we show the predicted VCHR distribution across different camera frames currently rendered on the MAR device. We observe that the predicted VCHR for User~\#1 increases with higher upload frequencies, with the average VCHR stabilizing around~$3$\,Hz. This stabilization occurs because higher upload frequencies enable timely uploading of camera frames for device pose estimation, allowing the edge server to accurately predict subsequent MAR device poses for annotation rendering. Although allocating more radio spectrum to support higher upload frequencies enhances user QoE, the marginal improvement diminishes at higher frequencies. As a result, the customized DT for individual MAR device can guide network resource management to identify the optimal trade-off between radio spectrum resource consumption and user QoE.  

    \begin{figure}[t]
        \centering
        \includegraphics[width=0.3\textwidth]{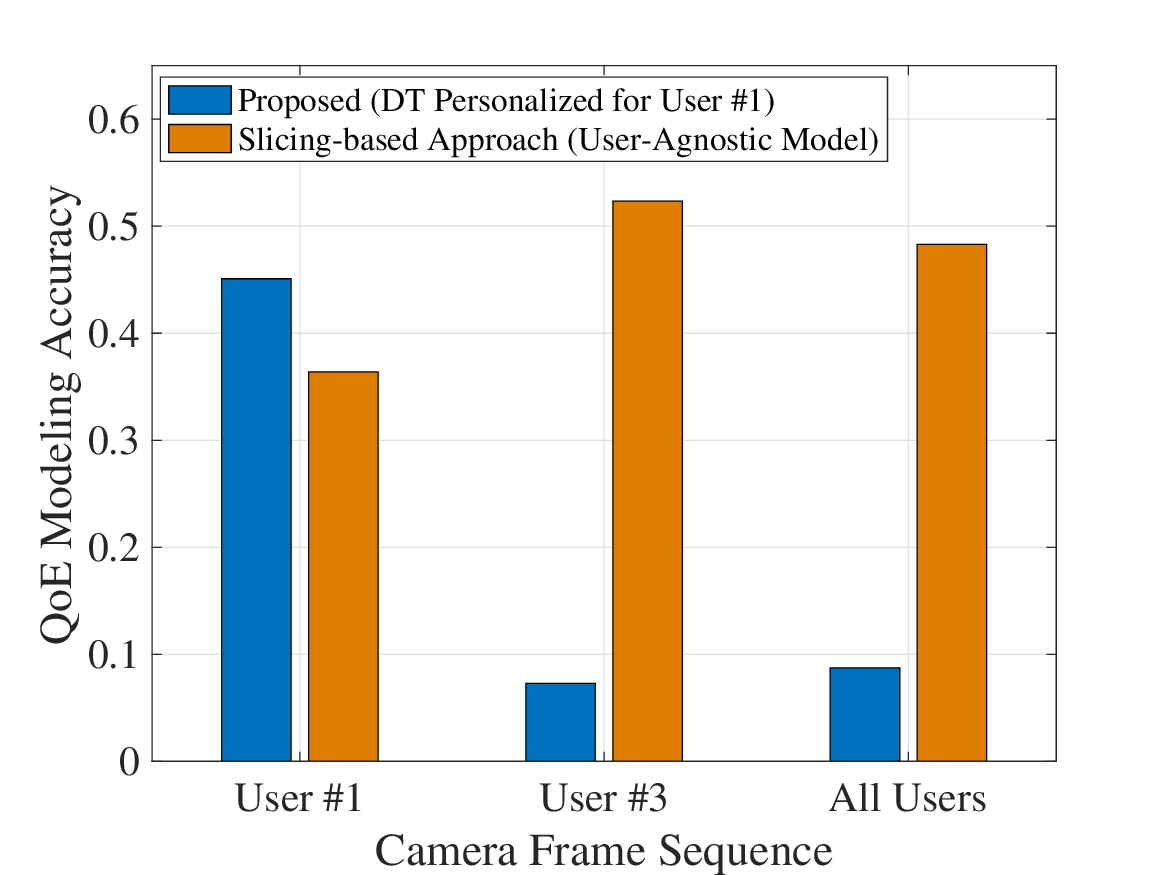}
        %\vspace{-0.3cm} % 调整负间距
        \caption{User QoE modeling accuracy comparison.}\label{fig2}
    \end{figure}

In Fig.~\ref{fig2}, we compare the QoE modeling accuracy of our DT-based approach with that of the network slicing-based approach employed in current 5G networks. Unlike the DT, which is customized to model the QoE for an individual user, the network slicing-based approach employs a user-agnostic QoE model trained on camera frame sequences aggregated from all users in MAR. The results show that the personalized DT for User~\#1 achieves higher QoE modeling accuracy on camera frame sequences from User~\#1 but performs worse on camera frame sequences from User~\#3 and the sequences aggregated from all users, compared to the slicing-based approach. This disparity arises because our DT-based approach focuses on data specific to an individual MAR device, effectively capturing differences in user movement and the surrounding physical environment. 

% the accuracy of two models in QoE modeling: a personalized DT (Digital Twin) model for User \#1, represented by blue bars, and a traditional network-slicing-based generic user model, represented by orange bars. The QoE accuracy of these two models is evaluated across three camera frame sequences. The figure shows that the personalized DT model for User \#1 outperforms the generic user model on the camera frame sequence of User \#1. This is because the DT model captures the unique impact of radio spectrum resource management on the VCHR of a single user. Compared to traditional user-agnostic methods, the DT model provides a higher level of granularity in QoE modeling.However, when validating the DT model on the camera frame sequences of User \#3 and the entire user set, its accuracy decreases. This is due to the lack of personalized training for these users, highlighting the dependency of the DT model on user-specific customization. For effective application, the DT model requires personalization for each individual user.In contrast, the traditional network-slicing-based generic user model does not consider individual differences. While it maintains relatively high accuracy for User \#3 and the entire user set, it performs worse than the personalized DT model on specific camera frame sequences tailored for a particular user.

\vspace{0.15cm}
\section{Conclusion}

% In this paper, we have developed a data-driven service provision approach based on the M-UDT to support customized user experiences in edge-assisted MAR. In the M-UDT, the established hierarchical data model organizes the factors affecting user-specific data traffic, and the designed UDT functions enable the switching between two data-driven traffic models to adapt to non-stationary data traffic. Simulation results have demonstrated the effectiveness of the developed M-UDT-based data-driven approach in reducing spectrum resource consumption while satisfying the delay requirement of camera frame uploading due to high modeling accuracy. Our approach provides a scalable and flexible paradigm to characterize the intricate impacts of MAR operational mechanisms on user-specific resource demands, which facilitates the shift to user-centric service provision in the 6G era. In the future, we plan to incorporate service provision for multiple MAR devices with diverse camera frame uploading mechanisms.

In this paper, we have proposed a digital twin (DT)-based communication service provision approach to optimize radio spectrum resource usage while ensuring user QoE, i.e., virtual content hit rate, for annotation rendering in edge-assisted mobile augmented reality (MAR). The developed DT, tailored for MAR applications, enables replicating key annotation rendering mechanisms and capturing the user-specific relationship between QoE and communication resource demands, improving the granularity of user QoE modeling. Leveraging these DTs, the proposed QoE-oriented resource allocation algorithm improves radio spectrum efficiency while satisfying diverse QoE requirements compared to conventional network slicing-based approaches used in 5G networks. This work has introduced a fine-grained paradigm for user QoE modeling, setting the stage for more personalized and efficient service provisioning for MAR rendering in 6G networks.
\vspace{0.15cm}
\appendix

\subsection{Proof of Theorem~\ref{theorem1}}\label{appendix:the1} 
\begin{proof}

We define random variable~$S = \sum_{f \in \mathcal{F}^\text{r}_{u}}{\mathbf{1}\left( h_{u,f}  \ge V_{u} \right)}$. The left-hand side of inequality~\eqref{eq7} can be re-written as follows: 
       \begin{equation}\label{}
            \hat{P}\left( \left(\frac{1}{|\mathcal{F}^\text{r}_{u}|} \sum_{f \in \mathcal{F}^\text{r}_{u}}{\mathbf{1}\left( h_{u,f}  \ge V_{u} \right)} \right) \ge \rho \right) = P (S \ge |\mathcal{F}^\text{r}_{u}| \rho ).
       \end{equation}
For each camera frame~$f \in \mathcal{F}^\text{r}_{u}$,~$\mathbf{1}\left( h_{u,f}  \ge V_{u} \right)$ is a Bernoulli random variable with mean~$\mathbb{E}[\mathbf{1}\left( h_{u,f}  \ge V_{u} \right)] = \hat{p}_{u,f}$. Random variable~$S$ is the sum of $|\mathcal{F}^\text{r}_{u}|$ independent but not identically distributed Bernoulli random variables. The mean and variance of $S$ are $\mathbb{E}[S] = \sum_{f \in \mathcal{F}^\text{r}_{u}}{\hat{p}_{u,f}}$ and $\sigma^{2} = Var(S) = \sum_{f \in \mathcal{F}^\text{r}_{u}}{\hat{p}_{u,f} (1 - \hat{p}_{u,f})}$, respectively.

When the number of camera frames, i.e.,~$|\mathcal{F}^\text{r}_{u}|$, is large, random variable~$S$ can be approximated by a normal distribution based on the central limit theorem as follows:
       % \begin{equation}\label{}
       %     P (S \ge |\mathcal{F}^\text{r}_{u}| \rho ) \approx P (Z \ge \frac{|\mathcal{F}^\text{r}_{u}| \rho - \mathbb{E}[S]}{\sigma}), 
       % \end{equation}
    % \begin{equation}\label{}
    %    \hat{P}\left( \left(\frac{1}{|\mathcal{F}^\text{r}_{u}|} \sum_{f \in \mathcal{F}^\text{r}_{u}}{\mathbf{1}\left( h_{u,f}  \ge V \right)} \right) \ge \rho \right) \approx   1 - \Phi \left( \frac{|\mathcal{F}^\text{r}_{u}| \rho -  \sum_{f \in \mathcal{F}^\text{r}_{u}}{\hat{p}_{u,f}} }{ \sqrt{\sum_{f \in \mathcal{F}^\text{r}_{u}}{\hat{p}_{u,f} (1 - \hat{p}_{u,f})}} } \right), 
    % \end{equation}

    \begin{subequations}\label{}
        \begin{align}
            P (S \ge |\mathcal{F}^\text{r}_{u}| \rho ) & \approx P (Z \ge \frac{|\mathcal{F}^\text{r}_{u}| \rho - \mathbb{E}[S]}{\sigma})\\
            & = 1 - \Phi \left( \frac{|\mathcal{F}^\text{r}_{u}| \rho -  \sum_{f \in \mathcal{F}^\text{r}_{u}}{\hat{p}_{u,f}} }{ \sqrt{\sum_{f \in \mathcal{F}^\text{r}_{u}}{\hat{p}_{u,f} (1 - \hat{p}_{u,f})}} } \right),
        \end{align}
    \end{subequations}
where~$Z$ represents a standard normal random variable, and~$\Phi(\cdot)$ is the cumulative distribution function of the standard normal distribution. Thus, \eqref{eq7} can be re-written as~\eqref{eq11}.

\end{proof}

\vspace{0.1cm}
\section*{Acknowledge}
This work was supported by the National Key Research and Development Program of China under Grant 2022YFB2901900.

% % \balance

% \section*{Acknowledgment}
% % \vspace{-0.1cm}
% \bl{This work is supported by the Natural Sciences and Engineering Research Council (NSERC) of Canada.}
% National Natural Science Foundation of China (NSFC) under Grant No. 91638204 and

\bibliography{ref_AR3}

\begin{thebibliography}{10}
\providecommand{\url}[1]{#1}
\csname url@samestyle\endcsname
\providecommand{\newblock}{\relax}
\providecommand{\bibinfo}[2]{#2}
\providecommand{\BIBentrySTDinterwordspacing}{\spaceskip=0pt\relax}
\providecommand{\BIBentryALTinterwordstretchfactor}{4}
\providecommand{\BIBentryALTinterwordspacing}{\spaceskip=\fontdimen2\font plus
\BIBentryALTinterwordstretchfactor\fontdimen3\font minus
  \fontdimen4\font\relax}
\providecommand{\BIBforeignlanguage}[2]{{%
\expandafter\ifx\csname l@#1\endcsname\relax
\typeout{** WARNING: IEEEtran.bst: No hyphenation pattern has been}%
\typeout{** loaded for the language `#1'. Using the pattern for}%
\typeout{** the default language instead.}%
\else
\language=\csname l@#1\endcsname
\fi
#2}}
\providecommand{\BIBdecl}{\relax}
\BIBdecl

\bibitem{cheng2024toward}
N.~Cheng, X.~Wang, Z.~Li, Z.~Yin, T.~Luan, and X.~Shen, ``Toward enhanced
  reinforcement learning-based resource management via digital twin:
  {Opportunities}, applications, and challenges,'' \emph{IEEE Netw.}, pp. 1--7,
  2024, to be published, doi:10.1109/MNET.2024.3438543.

\bibitem{chen2023adaptslam}
Y.~Chen, H.~Inaltekin, and M.~Gorlatova, ``{AdaptSLAM}: {Edge}-assisted
  adaptive {SLAM} with resource constraints via uncertainty minimization,'' in
  \emph{Proc. IEEE INFOCOM}, 2023, New York, NY, USA.

\bibitem{zhou2024user_wcm}
C.~Zhou, S.~Hu, J.~Gao, X.~Huang, W.~Zhuang, and X.~Shen, ``User-centric
  immersive communications in {6G}: {A} data-oriented approach via digital
  twin,'' \emph{arXiv:2411.05184}, [Online]. Available:
  https://doi.org/10.48550/arXiv.2410.02688.

\bibitem{huzaifa2021illixr}
M.~Huzaifa, R.~Desai, S.~Grayson, X.~Jiang, Y.~Jing, J.~Lee, F.~Lu, Y.~Pang,
  J.~Ravichandran, F.~Sinclair \emph{et~al.}, ``{ILLIXR}: {Enabling} end-to-end
  extended reality research,'' in \emph{Proc. IEEE IISWC}, 2021, Storrs, CT,
  USA.

\bibitem{ma2023nomore}
X.~Ma, Q.~Zeng, H.~Chi, and L.~Luo, ``No more companion {Apps} hacking but one
  dongle: {Hub-based} blackbox fuzzing of {loT} firmware,'' in \emph{Proc.~ACM
  MobiSys}, 2023, Helsinki, Finland.

\bibitem{hu2023adaptive}
S.~Hu, M.~Li, J.~Gao, C.~Zhou, and X.~Shen, ``Adaptive device-edge
  collaboration on {DNN} inference in {AIoT}: {A} digital twin-assisted
  approach,'' \emph{IEEE IoT J.}, vol.~11, no.~7, pp. 12\,893--12\,908, 2023.

\bibitem{ran2020multi}
X.~Ran, C.~Slocum, Y.-Z. Tsai, K.~Apicharttrisorn, M.~Gorlatova, and J.~Chen,
  ``Multi-user augmented reality with communication efficient and spatially
  consistent virtual objects,'' in \emph{Proc. ACM CoNEXT}, 2020, New York, NY,
  USA.

\bibitem{zhou2024digital}
C.~Zhou, J.~Gao, M.~Li, N.~Cheng, X.~Shen, and W.~Zhuang, ``Digital twin-based
  {3D} map management for edge-assisted device pose tracking in mobile {AR},''
  \emph{IEEE IoT J.}, vol.~11, no.~10, pp. 17\,812--17\,826, 2024.

\bibitem{linowes2017augmented}
J.~Linowes and K.~Babilinski, \emph{Augmented reality for developers: {Build}
  practical augmented reality applications with {Unity}, {ARCore}, {ARKit}, and
  {Vuforia}}.\hskip 1em plus 0.5em minus 0.4em\relax Packt Publishing Ltd,
  2017.

\bibitem{pan2024quality}
G.~Pan, S.~Xu, S.~Zhang, X.~Chen, and Y.~Sun, ``Quality of experience oriented
  cross-layer optimization for real-time {XR} video transmission,'' \emph{IEEE
  Trans. Circuits Syst. Video Technol.}, vol.~34, no.~8, pp. 7742--7755, 2024.

\bibitem{feng2023qoe}
J.~Feng, L.~Liu, X.~Hou, Q.~Pei, and C.~Wu, ``{QoE} fairness resource
  allocation in digital twin-enabled wireless virtual reality systems,''
  \emph{IEEE J. Sel. Areas Commun.}, vol.~41, no.~11, pp. 3355--3368, 2023.

\bibitem{ren2020edge}
P.~Ren, X.~Qiao, Y.~Huang, L.~Liu, C.~Pu, S.~Dustdar, and J.~Chen, ``{Edge}
  {AR} {X5}: {An} edge-assisted multi-user collaborative framework for mobile
  web augmented reality in {5G} and beyond,'' \emph{IEEE Trans. Cloud Comput.},
  vol.~10, no.~4, pp. 2521--2537, 2020.

\bibitem{campos2021orb}
C.~Campos, R.~Elvira, J.~J.~G. Rodr{\'\i}guez, J.~M. Montiel, and J.~D.
  Tard{\'o}s, ``{ORB-SLAM3}: {An} accurate open-source library for visual,
  visual--inertial, and multimap {SLAM},'' \emph{IEEE Trans. Robot.}, vol.~37,
  no.~6, pp. 1874--1890, 2021.

\bibitem{han2020vivo}
B.~Han, Y.~Liu, and F.~Qian, ``Vivo: {Visibility}-aware mobile volumetric video
  streaming,'' in \emph{Proc. ACM MobiCom}, 2020, New York, NY, USA.

\bibitem{ma2024from}
X.~Ma, L.~Luo, and Q.~Zeng, ``From one thousand pages of specification to
  unveiling hidden bugs: {Large} language model assisted fuzzing of matter
  {IoT} devices,'' in \emph{Proc. USENIX Security}, 2024, Philadelphia, PA,
  USA.

\bibitem{song2024canal}
E.~Song, Y.~Song, C.~Lu, T.~Pan, S.~Zhang, J.~Lu, J.~Zhao, X.~Wang, X.~Wu,
  M.~Gao \emph{et~al.}, ``Canal mesh: {A} cloud-scale sidecar-free multi-tenant
  service mesh architecture,'' in \emph{Proc. ACM SIGCOMM}, 2024, Sydney,
  Australia.

\end{thebibliography}

\bibliographystyle{IEEEtran}

\end{document}